\documentclass{ita}

\usepackage[noend]{algpseudocode}
\usepackage{algorithm}
\usepackage{times}
\usepackage{amssymb}
\usepackage{amsmath}
\usepackage{latexsym}
\usepackage{graphics}
\usepackage{url}
\usepackage{verbatim}
\usepackage{color}
\usepackage{setspace}
\usepackage{multirow}
\usepackage{lineno}
\usepackage{booktabs}
\usepackage{graphicx}
\usepackage{caption}
\usepackage[shortlabels]{enumitem}
\usepackage{diagbox}

\algnewcommand\And{\textbf{and}}
\algnewcommand\Or{\textbf{or}}
\algnewcommand\Not{\textbf{not}}
\algnewcommand\In{\textbf{in}}
\algnewcommand\Each{\textbf{each}}

\newtheorem{theorem}{Theorem}[section]          % Numbers theorems bysection.
\newtheorem{lemma}[theorem]{Lemma}             % Numbers a lemma by a theorem number.
\newcommand{\squishlist}{
 \begin{list}{$\bullet$}
  { \setlength{\itemsep}{0pt}
     \setlength{\parsep}{3pt}
     \setlength{\topsep}{3pt}
     \setlength{\partopsep}{0pt}
     \setlength{\leftmargin}{2.5em}
     \setlength{\labelwidth}{1em}
     \setlength{\labelsep}{0.5em} } }

\newcommand{\squishlisttwo}{
 \begin{list}{$\triangleright$}
  { \setlength{\itemsep}{0pt}
     \setlength{\parsep}{0pt}
    \setlength{\topsep}{0pt}
    \setlength{\partopsep}{0pt}
    \setlength{\leftmargin}{2em}
    \setlength{\labelwidth}{1.5em}
    \setlength{\labelsep}{0.5em} } }

\newcommand{\squishend}{
  \end{list}  }

\usepackage{listings}

\definecolor{verbgray}{gray}{0.9}

\lstnewenvironment{code}{%
  \lstset{backgroundcolor=\color{verbgray},
 frame=single,
  language=C,
  framerule=0pt,
  basicstyle=\ttfamily,
  showstringspaces=false,
  commentstyle=\color{blue}\textit,
  keywordstyle=\color{black}\bf,
  numbers=none,
  keepspaces=true,
  columns=fullflexible}}{}

\definecolor{shadecolor}{rgb}{.91, .91, .91}
\definecolor{bordercolor}{rgb}{.8, .8, .6}

\definecolor{lipicsGray}{rgb}{0.31,0.31,0.33}
\definecolor{lipicsBulletGray}{rgb}{0.60,0.60,0.61}
\definecolor{lipicsLineGray}{rgb}{0.51,0.50,0.52}
\definecolor{lipicsLightGray}{rgb}{0.85,0.85,0.86}
\definecolor{lipicsYellow}{rgb}{0.99,0.78,0.07}

\definecolor{ultramarine}{rgb}{0, 0.125, 0.376}

 \definecolor{arsenic}{rgb}{0.23, 0.27, 0.29}
 \definecolor{beige}{rgb}{0.96, 0.96, 0.86}

\definecolor{amber}{rgb}{1.0, 0.75, 0.0}
\definecolor{orange}{rgb}{1.0, 0.49, 0.0}
\definecolor{dandelion}{rgb}{0.94, 0.88, 0.19}

  \definecolor{indiagreen}{rgb}{0.07, 0.53, 0.03}
  \definecolor{huntergreen}{rgb}{0.21, 0.37, 0.23}

\newcommand{\blue}[1] {\textcolor{blue}{#1}}
\newcommand{\red}[1] {\textcolor{red}{#1}}

\newcommand{\defo}[1] {\emph{\textcolor{blue}{#1}}}

\definecolor{shadecolor}{rgb}{.9, .9, .9}

 \usepackage{framed}

    {\endMakeFramed}

    \newenvironment{frshaded*}{%
    \MakeFramed {\advance\hsize-\width \FrameRestore}}%
    {\endMakeFramed}
    
    \newcounter{examplecounter}
\newenvironment{exam}{
 \begin{frshaded*}
    \refstepcounter{examplecounter}%
    \noindent
  \textbf{Example \arabic{examplecounter}}%
  \quad
}{%
\end{frshaded*}
}

{\endMakeFramed}

\newenvironment{frshaded2*}{%
    \MakeFramed {\advance\hsize-\width \FrameRestore}}%
    {\endMakeFramed}

\newenvironment{result}{
 \begin{frshaded2*}
}{%
\end{frshaded2*}

}

{\endMakeFramed}

\newenvironment{frshaded3*}{%
    \MakeFramed {\advance\hsize-\width \FrameRestore}}%
    {\endMakeFramed}

\graphicspath{{graphics/}}

\definecolor{winered}{rgb}{0.5,0.2,0}

\usepackage{hyperref}
\hypersetup{ 
    colorlinks = true,
    allcolors={winered},
}

\usepackage{wrapfig}

\usepackage[normalem]{ulem}
\usepackage{changepage}
\usepackage{tabularx}
\usepackage{ragged2e}
\renewcommand{\tt} {\mathtt}
\renewcommand{\S} {\mathbf{S}}
\newcommand{\edit}[1]{\textcolor{black}{#1}}

\begin{document}

%\title{ Random generation of universal cycles and de Bruijn sequences} %\thanks{NSERC}

\title{ Las Vegas algorithms to generate universal cycles and de Bruijn sequences uniformly at random} %\thanks{NSERC}

\author{Joe Sawada}
\author{Daniel Gabrić}


%\subjclass{the ams subject classification}

\begin{abstract}
We present practical algorithms for generating universal cycles uniformly at random.  In particular, we consider universal cycles for shorthand permutations, subsets and multiset permutations, weak orders, and orientable sequences. 
Additionally, we consider de Bruijn sequences, weight-range de Bruin sequences, and de Bruijn sequences, with forbidden $0^z$ substring.  Each algorithm, seeded with a random element from the given set, applies a random walk of an underlying Eulerian de Bruijn graph to obtain a random arborescence (spanning in-tree).
%edge is selected in the de Bruijn graph.  
Given the random arborescence and the de Bruijn graph, a corresponding random universal cycle can be generated in constant time per symbol.  %We provide experimental results on the average cover time required to compute a random arborescence for each object. 
We present experimental results on the average cover time needed to compute a random arborescence for each object using a Las Vegas algorithm.
\end{abstract}

\maketitle

%=====================================================================
\section{Introduction} \label{sec:intro}
%=====================================================================

Let $\mathbf{\Sigma}_k(n)$ denote the set of all strings of length $n$ over the alphabet $\{0,1, \ldots, k{-}1\}$.  Let $\S$ denote a subset of $\mathbf{\Sigma}_k(n)$.
%The \defo{de Bruijn graph} of ${\bf S}$, denoted $G(\mathbf{S})$, is the directed graph where each vertex corresponds to a length-($n{-}1$) prefix or suffix of a string in ${\bf S}$, and
%there is a directed edge from vertex $u$ to vertex $v$ if $u=u_1u_2\cdots u_{n-1}$ and $v = u_2u_2\cdots u_n$. Each edge corresponds to a string $u_1u_2\cdots u_n$ in $\mathbf{S}$ and we label such an edge by $u_n$.   
\edit{The \defo{de Bruijn graph} of ${\bf S}$, denoted $G(\mathbf{S})$, is the directed graph where each vertex corresponds to a length-($n{-}1$) prefix or a length-($n{-}1$) suffix of a string in ${\bf S}$; for each string  $u_1u_2\cdots u_n$ in $\mathbf{S}$
there is a directed edge labeled $u_n$ from vertex $u=u_1u_2\cdots u_{n-1}$ to vertex $v = u_2u_2\cdots u_n$.  }
%there is a directed edge from vertex $u$ to vertex $v$ if $u=u_1u_2\cdots u_{n-1}$ and $v = u_2u_2\cdots u_n$. . 
For example, see Figure~\ref{fig:shiftG}. 
%Interestingly, the de Bruijn graph of ${\bf S}$ is similar to the notion of a Rauzy graph~\cite{Rauzy:1983}, which is typically defined with respect to strings. 

%

\begin{figure}[h]
    \centering
    \includegraphics[width=0.9\linewidth]{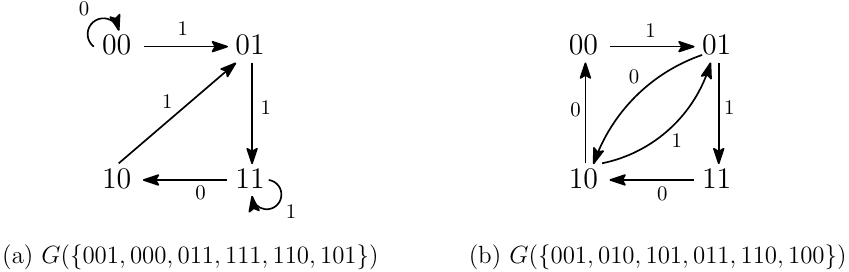}
    \caption{Two de Bruijn graphs. The graph in (b) is Eulerian. } 
    \label{fig:shiftG}
\end{figure}

A \defo{universal cycle} for $\S$, is a cyclic string of length $|\S|$ that contains each string in $\S$ as a substring exactly once (including the wraparound); they exist if and only if $G(\S)$ is Eulerian, that is, $G(\S)$ contains an Euler cycle. For example, consider $\mathbf{S} = \{001, 010, 101, 011, 110, 100\}$ and the de Bruijn graph $G(\mathbf{S})$ illustrated  in Figure~\ref{fig:shiftG}(b): the Euler cycle $00, 01, 11, 10, 01, 10, 00$ corresponds to the universal cycle $110100$ obtained by outputting the labels on the edges in the Euler cycle.  In this paper we are concerned with interesting subsets $\S$ whose underlying de Bruijn graph is Eulerian, i.e., $\S$ admits a universal cycle.  In particular, we consider:
\begin{enumerate}
    \item (shorthand) permutations, subsets, and permutations of a multiset,
    \item weak orders, 
    \item $k$-ary strings of length $n$ (which yield de Bruijn sequences),
    \item $k$-ary strings of length $n$  that do not contain $0^z$ (cyclically),
    \item$k$-ary strings of length $n$ with weight in the range $[a,b]$, and
    \item $k$-ary strings of length $n$ that produce asymptotically optimal orientable sequences, 
 %   \item Reverse Complement?
  %  \item Negative Orientable?
\end{enumerate}
where the \defo{weight} of a string is the sum of its symbols, and $[a,b]$ denotes the set of integers $\{a,a{+}1, \ldots, b\}$.

The primary objective of this paper is to describe practical algorithms to generate universal cycles for these objects \emph{uniformly at random}, while using exponential space to store an underlying de Bruijn graph.  \edit{We recall a generic Las Vegas algorithm from~\cite{KMUW96} that generates a random Euler cycle in any directed Eulerian graph by first generating a random arborescence (spanning in-tree).}  For each set $\S$, we generate a random string in $\S$ in order to seed the algorithm by selecting a root for the arborescence.  The generic algorithm is presented in Section~\ref{sec:algo}. Then, in Section~\ref{sec:app} we consider each of the aforementioned sets $\S$ and provide (i) a discussion of how to generate a random element from $\S$ and (ii) experimental evidence for the average cover time to compute a random arborescence in $G(\S)$.  
%We conclude with some directions for future research in Section~\ref{sec:fut}.

\medskip
\noindent
{\bf Motivation.}~~
This paper is motivated by a recent result from Lipt\'{a}k and Parmigiani~\cite{LP24} that generates random de Bruijn sequences, although not uniformly at random (see  Section~\ref{sec:DB}).  In that paper, they compared their approach with an implementation of Fleury’s algorithm~\cite{fleury} to generate Euler cycles, modified by adding randomization. That implementation could not generate all possible de Bruijn sequences, however, it served ``as the closest available method for comparison''~\cite{LP24}.   
Indeed, despite the vast literature on universal cycle constructions, and in particular, de Bruijn sequences, we also found no detailed discussion or resource regarding the generation of these sequences uniformly at random, other than a passing comment by Propp and Wilson in~\cite[p.172]{wilson2}.  
\edit{However, it is well known that a random arborescence in a directed Eulerian graph can be used to generate a random Euler cycle ~\cite{KMUW96}. }

%A recent result in~\cite{LP24} describes how to construct random de Bruijn sequences, although not uniformly at random. 
%\red{cite Fleury somewhere?} \red{elaborate on [15] as the key motivator}.

%=====================================================================
\section{Preliminaries} \label{sec:prelim}
%=====================================================================

Let $G = (V,E)$ denote a directed graph consisting of a non-empty set of vertices $V$ and an edge set $E$ consisting of ordered pairs of elements in $V$.  A \defo{walk} in $G$ is a sequence of vertices $v_1,v_2, \ldots, v_j$ such that $(v_i, v_{i+1}) \in E$ for all $i$ in $\{1,2, \ldots, j{-}1\}$.  A \defo{reverse walk} is a sequence of vertices $v_1,v_2, \ldots, v_j$ such that $(v_{i+1}, v_{i}) \in E$ for all $i$ in $\{1,2, \ldots, j{-}1\}$.  Let $G$ be represented by a standard adjacency list representation.  We define a \defo{traversal} to be a walk starting from some vertex $v$ that follows the rule: at each vertex $u$ , the next vertex corresponds to the first unused edge on $u$'s adjacency list; the traversal terminates when it reaches a vertex whose adjacency list has been exhausted.  

\begin{exam} \label{exam:graph} \small
Consider the directed graph $G=(V,E)$ where $V = \{u,v,w\}$ and $E = \{(u,u), (u,v), (v,w), (w,u)\}$.  Then $u,v,w$ is a walk in $G$, and $w,v,u$ is a reverse walk in $G$.  Give the adjacency list representation $u \rightarrow u,v$, $v \rightarrow w$, and $w \rightarrow u$, the traversal starting at $u$ is the walk $u,u,v,w,u$ using all four edges.  It corresponds to an Euler cycle in $G$.
\end{exam}

A traversal that starts and ends with the same vertex, say $r$, and visits all the edges in $E$ generates an Euler cycle, i.e., the traversal does not ``burn bridges''~\cite{fleury}. This means that the $|V|-1$ edges corresponding to the last edges on each adjacency list, except for $r$'s, form an arborescence  (spanning in-tree) rooted at $r$.  Using this well-known fact, all Euler cycles can be generated as follows:

%=========
\begin{result}  
\noindent {\bf Generate all Euler Cycles for a directed graph $G$}  \small

\smallskip

\begin{itemize}
    \item[(1)]  Generate all arborescences $\mathcal{T}$ for each possible root $r \in V$.
    \item[(2)]  For each $\mathcal{T}$ generated in step (1), take each edge $(u,v)$ in $\mathcal{T}$, and set the vertex $v$ to be the last on $u$'s adjacency list; then generate all possible orderings for the remaining vertices on each adjacency list.  
    \item[(3)]  For each set of adjacency lists generated in step (2), generate a traversal of $G$ starting at the corresponding root $r$. 
 \end{itemize}   
\end{result}
It is important to note that the above algorithm will generate all Euler cycles in $G$ exactly once, where the starting edge in each cycle is important.  By fixing a single root $r$ at step (1), we generate all Euler cycles up to equivalence when the edges are considered to be a circular list of edges; the starting edge in the cycle is immaterial.  However, note that the same equivalent cycle could be generated twice.  For instance, consider the graph in Example~\ref{exam:graph} with root vertex $u$. The algorithm produces two sets of edge listings at step (2): one where $u \rightarrow u,v$, and one where $u \rightarrow v,u$.  The two listings produce equivalent Euler cycles, namely $u,u,v,w,u$ and $u,v,w,u,u$.

%=====================================================================
\section{Random Generation} \label{sec:algo}
%=====================================================================

In this section the algorithm outlined in Section~\ref{sec:prelim} to generate all possible Euler cycles in a directed graph is applied to generate a single universal cycle \emph{uniformly at random} for a set $\mathbf{S}$ with an underlying Eulerian de Bruijn graph $G(\S)$. The Las Vegas algorithm describe by Algorithm R below summarizes the approach from Kandel et al.~\cite{KMUW96}, which extends the work from~\cite{alts}.  
We note one difference in our presentation. The algorithm in~\cite{KMUW96}, takes as input a (cyclic) sequence $\mathcal{S}$ that may contain duplicate length $n$ substrings.  From this sequence, it constructs a de Bruijn graph that allows for multiple edges between vertices.  The algorithm is then initialized by selecting a random rotation of $\mathcal{S}$, which effectively generates a random edge in the underlying graph.  For our purposes, we do not have an initial sequence $\mathcal{S}$. Instead, for each set $\mathbf{S}$ considered in  Section~\ref{sec:app}, we present an efficient algorithm to compute a random element in $\mathbf{S}$, which corresponds to a random edge in $G(\mathbf{S})$.

%\edit{Since there is at most one directed edge between two vertices in a de Bruijn graph $G(\S)$, the adjacency list of given vertex can be encoded by the labels of each outgoing edge. Thus, the graph in Figure~\ref{fig:shiftG}(b) has adjacency lists: $00 \rightarrow 1$, $01 \rightarrow 0,1$, $10 \rightarrow 0,1$, $11 \rightarrow 0$.}

%=========
\begin{result}  
\noindent {\bf Algorithm R} \small

\smallskip

\noindent
Generate a universal cycle for set $\S$ uniformly at random given the underlying (Eulerian) de Bruijn graph $G(\S)$:
\begin{enumerate}
       \item  Generate a random edge ($r,v$) in $G(\S)$, i.e., a string in $\mathbf{S}$, to obtain a random root vertex $r$ 
        \item Generate a random arborescence $T$ directed to root $r$
       \item Make each edge of $T$ (the bridges) the last edge on the adjacency list of the corresponding outgoing vertex (the root does not have such a bridge), then randomly assign the order of the remaining outgoing edges 
       \item Starting from $r$, perform a traversal of $G(\S)$, outputting the label on each edge as it is visited. 
       \end{enumerate}       
\end{result}

Selecting a random vertex instead of an edge at step (1) of Algorithm R will not lead to Euler cycle generated \emph{uniformly} at random if $G(\S)$ is not regular, as illustrated in the following example.
\begin{exam} \small
Consider $\mathbf{S} = \{001, 010, 101, 011, 110, 100\}$ and its corresponding de Bruijn graph $G(\mathbf{S})$ illustrated in Figure~\ref{fig:shiftG}(b).  Every vertex in $G(\mathbf{S})$ roots two unique spanning arborescences.  For each tree rooted at $r=000$ there is one corresponding edge labeling; however, each tree rooted at $r=010$ has two edge labelings since the there are two possible adjacency lists for the root.  Thus, when $r=000$, there are two possible universal cycles, while if $r=010$ there are 4 possible universal cycles that could get generated.  Thus, if a random vertex instead of a random edge is chosen at step (1) in Algorithm R, a universal cycle generated from $r=000$ will occur twice as frequently as a universal cycle rooted at $r=010$.
\end{exam}

The following example highlights the steps from Algorithm R.

%=========
\begin{exam} \small
Consider the set $\S$ consisting of all binary strings of length $n=6$ with weight in the range [1,2]. The de Bruijn graph $G(\S)$ is shown below with the randomly selected edge $\red{000100}$ from step (1) of Algorithm R.

%\begin{figure}[h]
  %  \centering
  \begin{center}
    \includegraphics[width=0.6\linewidth]{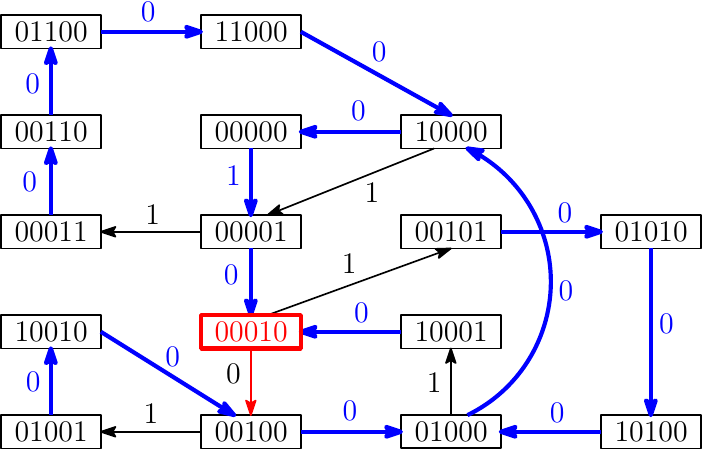}
  \end{center}
 %   \caption{The de Bruijn graph $G(\mathbf{S})$ for the set $\mathbf{S}$ consisting of all binary strings of length $n=6$ with weight in the range $[1,2]$. Steps (1) and (2) of Algorithm R are illustrated where \red{$000101$} is selected as a random edge, and the bold (blue) edges highlight a spanning in-tree rooted at $00010$. } 
  %  \label{fig:subset}
%\end{figure}

\noindent
The graph illustrates a random arborescence $T$ rooted at $r=\red{00010}$. Based on $T$, the ordering of the adjacency lists for the vertices following step (3) in Algorithm R is as follows, where the adjacency list of the root $00010$ is randomly selected as 00100,00101: 
\medskip
\begin{center}
\begin{tabular}{cc l} 
  00000 &$\rightarrow$  & \blue{00001} \\
  00001 &$\rightarrow$  & 00011, \blue{00010}  \\
  \red{00010} &$\rightarrow$  & 00100, 00101    \\
  00100& $\rightarrow$  & 01001, \blue{01000}    \\
  01000& $\rightarrow$  & 10001, \blue{10000}    \\ 
  10000 &$\rightarrow$  & 00001, \blue{00000}   \\
  00011& $\rightarrow$  & \blue{00110}  \\
  00101 & $\rightarrow$   & \blue{01010}    \\
 \end{tabular} \ \ \ \ 
 \begin{tabular}{cc l} 
  00110 &$\rightarrow$  & \blue{01100}   \\
  01001& $\rightarrow$  & \blue{10010}    \\
  01010& $\rightarrow$  & \blue{10100}    \\ 
  01100& $\rightarrow$  & \blue{11000}   \\
  10001& $\rightarrow$  & \blue{00010}    \\
  10010& $\rightarrow$  & \blue{00100}    \\
  10100 &$\rightarrow$  & \blue{01000}    \\ 
  11000& $\rightarrow$  & \blue{10000}.   \\
\end{tabular}
\end{center}
\bigskip

\noindent
Step (4) from Algorithm R produces the following random universal cycle for $\S$:
\[ 0100010100001100\red{00010}. \]
%

%Observe this is also a universal cycle for the 2-subsets from a ground set of size $n=7$; for each length $6$ substring add a 0 or 1 so the resulting length 7 string has exactly two 1s\cite{fixed-weight-binary-strings}.  

    \vspace{-0.1in}
\end{exam}

For Eulerian graphs, an arborescence can be generated uniformly at random by performing a random backwards walk until every vertex is visited. The first time a vertex is visited, the edge is recorded as a \defo{tree edge} in the random arborescence~\cite{KMUW96}.  The running time of this step depends on the \defo{cover time} of the random walk, which is the number of steps it takes to visit every vertex.  
%Thus, Algorithm R %has exponential time delay due to steps (2) and (3) and requires exponential space to store $G(\S)$. 
%requires exponential space, with respect to $n$ and $k$, and also requires \edit{expected} exponential time before the first symbol can be output.
\edit{The algorithm requires $\Theta(|\S|)$ space to store the graph.}
The final step (4) of Algorithm R generates a random universal cycle for $\S$ in constant time per symbol.

%\resizebox{2in}{!}{\includegraphics{subset1.pdf}}

%these are undirected graphs

%Arbitrary graphs have an $\mathcal{O}(n^3)$~\cite{FeigeU:1995} and $\Omega(n\log n)$~\cite{FeigeL:1995} expected cover time, and these bounds are tight, as witnessed by the lollipop graph with $n$ nodes and the complete graph with $n$ nodes respectively. 

%The expected cover time for random regular graphs is $\Theta(n\log n)$ with high probability~\cite{Cooper&Frieze:2012}. To the best of the authors' knowledge, the expected cover time of de Bruijn graphs remains unknown.

%The expected cover time of a \edit{random} directed graph is $\Omega\left(2^{|V|}\right)$ in the worst case~\cite{Cooper&Frieze:2012}. 
Theorem 3 in~\cite{KMUW96} shows that the expected cover time of a directed Eulerian graph is at most $|V|^2|E|$. %Thus, the expected cover time for the graphs in this paper is at most cubic in the number of vertices. 
\edit{In the next section, however, our experiments seem to indicate that some de Bruijn graphs appear to have a significantly faster expected cover time. We do not have a proof of this observation, nor have we found a proof in the literature. }

%\begin{conjecture}
%    The cover-time of a de Bruijn graph $G=(V,E)$ is $O(V \log V)$.
%\end{conjecture}

There are algorithms to generate random arborescences that are not dependent on the expected cover time. Wilson~\cite{wilson1} shows how to produce a random arborescence in time proportional to the maximum hitting time. The maximum hitting time is the maximum over all pairs of vertices $u$ and $v$ of the expected number of steps for a random walk to travel from $u$ to $v$.
Propp and Wilson~\cite{wilson2} show how to produce a random arborescence in time proportional to the smaller of the maximum hitting time and the mean hitting time. The mean hitting time is the average over all pairs of vertices $u$ and $v$ of the expected number of steps for a random walk to travel from $u$ to $v$.

%\resizebox{4in}{!}{\includegraphics{subset2.pdf}}

%=====================================================================
\section{Applications and experimental results}  \label{sec:app}
%=====================================================================

In this section we apply Algorithm R to generate universal cycles uniformly at random for
shorthand permutations, subsets and multiset permutations, $k$-ary strings (de Bruijn sequences), generalizations of de Bruijn sequences including those with no $0^z$ substring and those with bounded weight, and orientable sequences.

For each object, we discuss how to generate a random edge in the underlying de Bruijn graph to seed the algorithm, and present experimental evidence for the cover time required to compute a random arborescence.
Implementations of our algorithms are available at \url{http://debruijnsequence.org/db/random}.  In our implementations to compute the random arborescences, we did not pre-compute the shift-graphs, but instead used a mapping of each vertex to an integer (using a ranking algorithm or similar) to store whether a vertex had been visited.  We applied a similar strategy to generate the random universal cycle in step (4) of Algorithm R.

The results presented in this section with respect to the cover times of certain de Bruijn graphs are experimental.  We leave it as in interesting open problem to determine a tight upper bound on the expected cover time for the de Bruijn graphs being considered.

%=====================================================================
\subsection{Permutations, subsets, and multiset permutations}  \label{sec:perm}
%=====================================================================

Universal cycles do not exist, in general, for permutations and subsets.  For permutations, however, observe that the final symbol is redundant.  If
$p_1p_2\cdots p_n$ is a permutation, we say that 
$p_1p_2\cdots p_{n-1}$ is a \defo{shorthand permutation} of order $n$, where the last symbol is implied.
Let $\mathbf{SP}(n)$ be the set of all shorthand permutations of order $n$. Similarly, if $b_1b_2\cdots b_n$ is a binary string with $k$ ones (representing a $k$-subset of an $n$-set), we say that $b_1b_2\cdots b_{n-1}$ is a \defo{shorthand $k$-subset} of order $n$, where the last bit is implied.  Let $\mathbf{S}(n,k)$ be the set of all shorthand $k$-subsets of order $n$.
Multiset permutations (strings with fixed content) generalize both permutations and subsets. If $m_1m_2\cdots m_n$ is a permutation of the multiset $\{s_1, s_2, \ldots, s_n\}$, then
we say $m_1m_2\cdots m_{n-1}$ is a \defo{shorthand multiset permutation}.
When each $s_i = i$, a multiset permutation is simply a permutation, and when the multiset contains $k$ ones and $(n{-}k)$ zeros, it represents a binary string with weight $k$ representing a $k$-subset of an $n$-set.

For shorthand permutations, the underlying de Bruijn graph has $n!$ edges. Each vertex has in-degree = out-degree = $2$ and thus there are $n!/2$ vertices.  For shorthand $k$-subsets, where $k \geq 2$, the underlying de Bruijn graph has ${n \choose k}$ edges and each vertex is a binary string of length $n{-}2$ with weight $k-2$, $k-1$, or $k$; there are ${n-2 \choose k-2} + {n-2 \choose k-1} + {n-2 \choose k}$ vertices. Note that each vertex has the same in-degree as out-degree; however, this value may be either $1$ or $2$. For example, if $n=5$ and $k=2$, the vertex $001$ has  two incoming edges from $000$ and $100$, and outgoing edges to $010$ and $011$, while the vertex $000$ has one incoming edge from $100$ and one outgoing edge to $001$.

It is well known that a random permutation can be generated  by applying the Fisher-Yates shuffle~\cite{FisherYates}; an $O(n)$ time implementation is provided by Knuth~\cite[Algorithm P]{knuth2} based on a presentation by Durstenfeld~\cite{durstenfeld}. As noted by Arndt~\cite{arndt}, it is straightforward to apply the shuffle to generate an unbiased random multiset permutation $m_1m_2\cdots m_n$ in $\mathcal{O}(n)$ time as illustrated in Algorithm~\ref{alg:perm}.
%A random multiset permutation $m_1m_2\cdots m_n$ can be generated in $\mathcal{O}(n)$ time by applying a variant of the \emph{Fisher-Yates shuffle} as outlined by Arndt~\cite{Arndt}, and illustrated in Algorithm~\ref{alg:perm}.  The Fisher-Yates shuffle~\cite{FisherYates}, also known as the Knuth Shuffle~\cite{Knuth2}
%It is well-known that a random permutation can be generated in $\mathcal{O}(n)$ time by the \emph{Fisher-Yates shuffle}~\cite{FisherYates}\footnote{Also known as the \emph{Knuth shuffle} following his presentation of Algorithm P in~\cite{knuth2} based on the form of Durstenfeld~\cite{durstenfeld}.  A , as  based on the presentation by Arndt~\cite{arndt}. 
Thus, $m_1m_2\cdots m_{n-1}$ is a random shorthand multiset permutation that can be used to seed Algorithm R for shorthand permutations, shorthand subsets, and more generally, shorthand multiset permutations.
%
%======================
\begin{algorithm} 

\small
\begin{algorithmic}
\State $m_1m_2\cdots m_n \gets s_1s_2\cdots s_n$
\For{$i$ {\bf from} $n$ {\bf down to} $2$} 
    \State $j \gets $ random integer in $[1,i]$
    \State \Call{Swap}{$m_i, m_j$}
\EndFor
\end{algorithmic}
\caption{Random generation of a permutation $m_1m_2\cdots m_n$ of the multiset $\{s_1,s_2, \ldots, s_n\}$ applying the Fisher-Yates shuffle}
\label{alg:perm}
\end{algorithm}
%======================
%

   Table~\ref{tab:perm} shows the minimum, maximum, and average ratios of the cover time to total edges in $G(\mathbf{SP}(n))$ and $G(\mathbf{S}(n,n/2))$ by running Algorithm R for 10,000 iterations.  

\begin{table}[h] 
\begin{center}
\begin{tabular} {r |  r  r   r }
	& \multicolumn{3}{c}{\footnotesize {\bf Ratio}: cover time / $n!$} \\ %(10,000 iterations)} \\
 $n$ & ~~~~~~~~~{\bf Min} & ~~~~~~~~~{\bf Max}  &~~~~~~~~ {\bf Avg}   \\  \hline
3 & 0.3  & 0.3  & 0.3 \\
 4 & 0.5 & 3.7 & 1.1 \\
 5 & 0.8 & 5.3 & 2.0\\
 6 & 1.7 & 6.8 & 3.0\\
 7 & 2.7 & 6.5 & 3.9\\
 8 & 3.9 & 8.0 & 4.9\\
 9 & 5.0 & 7.7 & 5.9\\   %~one hour
 10 & 6.4 & 8.1 & 7.0 \\ %13 hour
\end{tabular} \ \ \ \ \ \ \ \  \ \ \ \ \ \ \ \ \
\begin{tabular} {r |  r  r   r }
	& \multicolumn{3}{c}{\footnotesize {\bf Ratio}: cover time / ${n \choose k}$} \\ %(10,000 iterations)} \\
 $n$ & ~~~~~~~~~{\bf Min} & ~~~~~~~~~{\bf Max}  &~~~~~~~~ {\bf Avg}   \\  \hline
10 & 1.8	&  18.9	& 5.0\\
 12 & 3.2	& 18.0 &6.3 \\
 14 &  4.6	& 16.4 &7.6 \\  
 16 &  5.9	& 18.8 & 8.8 \\ 
 18 &  7.2 & 20.3 & 10.1 \\  
 20 &  8.5 & 21.8  & 11.4  \\  %32min  
 22 & 9.7 & 21.8 & 12.7  \\  %~2h
 24 &  11.2& 22.6 & 14.1  \\    %~14h ?
 26 &  12.4&  24.0 & 15.4 \\    %~75h 
\end{tabular}
\end{center}
\caption{The minimum, maximum, and average ratios of the cover time to total edges in the de Bruijn graphs $G(\mathbf{SP}(n))$ (left) and $G(\mathbf{S}(n,n/2))$ (right) by running Algorithm R for 10,000 iterations.}
\label{tab:perm}
\end{table}

Universal cycles for shorthand permutations can be constructed in $\mathcal{O}(1)$ amortized time per symbol using $\mathcal{O}(n)$ space~\cite{shorthand2,shorthand}.
An $\mathcal{O}(n)$-time successor rule is presented in~\cite{karyframework}, and 
an $\mathcal{O}(1)$ amortized time per symbol algorithm applying concatenation trees
is presented in~\cite{concattree} that uses $\mathcal{O}(n^2)$ space.  
Universal cycles for shorthand subsets can be constructed in $\mathcal{O}(1)$ amortized time per symbol using $\mathcal{O}(n)$ space~\cite{fixed-weight-binary-strings}.
Universal cycles for shorthand multiset permutations (strings with fixed content) can be constructed in $\mathcal{O}(1)$ amortized time per symbol using $\mathcal{O}(n)$ space~\cite{fixed-content}.

%=====================================================================
\subsection{Weak orders}  \label{sec:weak}
%=====================================================================

    A \defo{weak order} is the number of ways $n$ competitors can finish in a race if ties are allowed.  Let $\mathbf{W}(n)$ denote the set of weak orders with $n$ competitors, and let $W_n$ denote $|\mathbf{W}(n)|$.   For example, \[\mathbf{W}(3) = \{111, 113, 131, 311, 122, 212, 221, 123, 132, 213, 231, 312, 321\},\] and $W_3 = 13$.  
    The number of weak orders where there is a $k$-way tie for first is given by ${n \choose k} W_{n-k}$ for $k<n$; there is 1 weak order when $k=n$.  Thus, $W_n = \sum_{k=1}^n {n \choose k} W_{n-k}$.  

    To generate a random weak order, we will apply this recurrence to group the strings of $\mathbf{W}(n)$ based on their content.
    %Applying this recurrence, we can apply a partial ordering to $\mathbf{W}(n)$ based on the content (Parikh vector) of the strings. 
    Let $c_i$ denote the number of occurrences of the symbol $i$ in $\omega = w_1w_2\cdots w_n$. Since every weak order contains 1 as its smallest symbol and the largest symbol possible is $n$, we say that $(c_1, c_2, \ldots, c_n)$ is the \defo{content} (also known as the Parikh vector) of $\omega$.  For example, if $n=3$ then each weak order has content $(1,1,1), (1,2,0), (2,0,1),$ or $(3,0,0)$.  If we partition $\mathbf{W}(n)$ into subsets based on their content, we can order the subsets based on the lexicographic ordering of the corresponding content. 
    The weak orders with exactly one $1$ comes first, followed by those with exactly two $1$s and so on. 
    %Then each grouping is ordered in similar way by the number of 2s, then those groups are similarly ordered by the number of 3s, and so on.  For $n=3$ we obtain the partial ordering $\{123,132, 213,231,312,321\}, \{122, 212, 221\}, \{113, 131,311\}, \{111\}$ of the 13 weak orders, where each grouping has exactly the same content. In other words, each group corresponds the permutations of some multiset.
    %
    Thus, to generate a random weak ordering in $\mathbf{W}(n)$, we can (i) select a random integer $r$ in $[1,W_n]$, (ii) apply the recurrence to determine the content $(c_1,c_2, \ldots, c_n)$ based on the described ordering, and (iii) apply a Fisher-Yates shuffle to obtain a random multiset permutation (weak order) with content $(c_1,c_2, \ldots, c_n)$ (see Section~\ref{sec:perm}).  Details are provided in Algorithm~\ref{alg:weak}.  Note that the integer $r$ does not uniquely determine the weak order being generated, however, it is possible to obtain this property by applying an unranking algorithm for multiset permutations instead of applying the shuffle. 
    
   % and selecting a random integer in the range $[1,W_n]$, we can compute the number of occurrences of the value 1 in a random weak order $w_1w_2\cdots w_n$.

   % This recurrence together with the Fisher-Yates shuffle can be applied to generate a weak order $w_1w_2\cdots w_n$ uniformly at random, as illustrated in Algorithm~\ref{alg:weak}. The strategy behind the algorithm at each iteration of the {\bf while} loop applies two steps for placing values $v$ into the $t$ remaining unassigned positions of $w_1w_2\cdots w_n$. First, by selecting a random number in the range $[1,W_t]$ and applying the recurrence for $W_t$, we  determine the number of elements $j$ with value $v$.
   % The second step randomly assigns the $j$ elements into the $t$ remaining spots by applying the Fisher-Yates shuffle. This is done by first generating a random binary string of length $t$ with $j$ ones.  The ones indicate the relative positions to assign the value $v$ into the remaining $t$ positions. Once the $j$ elements with value $v$ have been assigned, the number of remaining unassigned positions is $t-j$ and the next possible value is $v+j$, since there was a $j$-way tie at position $v$.

    %
%======================
\begin{algorithm}
\small
    \begin{algorithmic}

   % \State $w_1w_2\cdots w_n \gets 0^n$
    \State $(c_1,c_2, \ldots, c_n) \gets (0,0, \ldots 0)$
    \State $t \gets n$
    \State $ v \gets 1$
    %\State  \blue{\mbox{$\triangleright$~Randomly place the element(s) $v$ based on the recurrence}}
    \State $r \gets$ random integer in $[1,W_n]$ 
    \While{$t\geq 1$} 
    \State \blue{\mbox{$\triangleright$ ~Determine $c_v$ applying the recurrence for $W_t$}}
 
        \For{$j$ {\bf from} $1$ {\bf to} $t$} 
            \State $p_j \gets {t \choose j} W_{t-j}$
            \If{$r \leq p_j$} {\bf break} \EndIf
            \State $r \gets r -p_j$
        \EndFor
        \State $c_v \gets j$
        \State $v \gets v+j$
        \State $t \gets t-j$
     \EndWhile
    \Statex
     \State \blue{\mbox{$\triangleright$ ~Apply Fisher-Yates shuffle to generate a random multiset permutation}}
     \State $\omega \gets$ a random multiset permutation with content $(c_1,c_2,\ldots, c_n)$

      %  \State \blue{\mbox{$\triangleright$ ~Randomly place $j$ occurrences of $v$ into the $t$ empty spots}}
      %  \State $b_1\cdots b_t \gets$ a random binary string with length $t$ and weight $j$ \ \ \blue{\mbox{$\triangleright$ ~Apply Fisher-Yates shuffle}}
        
        %\State $b_1\cdots b_t \gets 1^j0^{t-j}$
       % \For{$k$ {\bf from} $t$ {\bf down to} $2$}
        %    \State $i \gets$ random integer in  $[1,k]$ 
         %   \State \Call{Swap}{$b_i,b_k$}
       % \EndFor
   %     \State $i \gets 1$
   %     \For{$k$ {\bf from} $1$ {\bf to} $n$} 
    %        \If {$w_k = 0$ }
   %             \If{$b_{i} = 1$}  \    
   %             $w_k \gets  v$; 
   %             \EndIf
   %             \State $i\gets i+1$ 
   %         \EndIf
   %     \EndFor
   %     \State $v \gets v+j$
   %     \State $t \gets t-j$
  
\end{algorithmic}
\normalsize
\caption{Random generation of a weak order $\omega$ from $\mathbf{W}(n)$.}
\label{alg:weak}
\end{algorithm}
   %======================

    \begin{lemma}
    Algorithm~\ref{alg:weak} can generate a weak order $w_1w_2\cdots w_n$ uniformly at random using $\mathcal{O}(n^2)$ simple operations on numbers up to $W_n$.
    \end{lemma}
    \begin{proof}
    Consider Algorithm~\ref{alg:weak}. 
    The values $W_j$, for $1 \leq j \leq n$ can be precomputed via dynamic programming using $\mathcal{O}(n^2)$ simple operations on numbers up to $W_n$.  We can compute the required binomial coefficients of the form ${n \choose k}$ in the same time using Pascal's identity. The {\bf while} loop iterates at most $n$ times and each iteration requires at most $\mathcal{O}(n)$ simple operations on numbers up to $W_n$.
    \end{proof}

    Table~\ref{tab:weak} shows the minimum, maximum, and average ratios of the cover time to total edges in $G(\mathbf{W}(n))$ by running Algorithm R for 10,000 iterations.
 \begin{table}
\begin{center}
\begin{tabular} {r |  r  r   r }
	& \multicolumn{3}{c}{\footnotesize {\bf Ratio}: cover time / $W_n$} \\ % (10,000 iterations)} \\
 $n$ & ~~~~~~~~~{\bf Min} & ~~~~~~~~~{\bf Max}  &~~~~~~~~ {\bf Avg}   \\  \hline
 3 & 0.5 & 9.3& 1.8 \\
 4 & 1.1 & 12.7& 3.7 \\
 5 & 2.5 & 12.9 & 5.4\\
 6 & 4.1 & 16.7 & 7.2\\
 7 & 6.4 & 21.3 & 9.3\\  % 40 min
 8 & 8.7 & 18.5 & 11.5\\   %8 hour
 9 & 12.0 & 15.4 & 13.8\\     % 13840m   (10 days)
\end{tabular}
\end{center}
\caption{The minimum, maximum, and average ratios of the cover time to total edges in $G(\mathbf{W}(n))$ by running Algorithm R for 10,000 iterations.}
\label{tab:weak}
\end{table}
%================

Universal cycles for weak orders can be constructed via a successor rule that generates the sequence in $\mathcal{O}(n)$ time per symbol using $\mathcal{O}(n)$ space~\cite{weakorder}.  By applying concatenation trees, they can be generated in $\mathcal{O}(1)$ amortized time using $\mathcal{O}(n^2)$ space~\cite{concattree}.  See the enumeration sequence A000670 for $W_n$  in the Online Encyclopedia of Integer Sequences~\cite{oeis1}.

%=====================================================================
\subsection{de Bruijn sequences}  \label{sec:DB}
%=====================================================================

For de Bruijn sequences, the underlying de Bruijn graph $G(\mathbf{\Sigma}_k(n))$ has $k^{n-1}$ vertices and $k^n$ edges.  A random $k$-ary string (edge) can be computed in $\mathcal{O}(n)$ time by generating a random symbol in $[0,k-1]$ $n$ times.
Table~\ref{tab:DB} and Table~\ref{tab:DB2} show the minimum, maximum, and average ratios of the cover time to total edges in $G(\mathbf{\Sigma}_k(n))$, for $k=2,3,4$ by running Algorithm R for 10,000 iterations.  

%==============
\begin{table}
\begin{center}
\begin{tabular} {r |  r  r   r }
	& \multicolumn{3}{c}{\footnotesize {\bf Ratio}: cover time / $2^n$ } \\ %(10,000 iterations)} \\
 $n$ & ~~~~~~~~~{\bf Min} & ~~~~~~~~~{\bf Max}  &~~~~~~~~ {\bf Avg}   \\  \hline
 4 &  0.4 &  8.4 &  1.3 \\ 
 5 &  0.5 &  7.1 &  1.7 \\ 
 6 &  0.6 &  8.3 &  2.1 \\ 
 7 &  0.9 &  8.7 &  2.4 \\ 
 8 &  1.1 &  8.6 &  2.7 \\ 
 9 &  1.4 &  9.0 &  3.1 \\ 
10 &  1.9 &  9.0 &  3.4 \\ 
11 &  2.2 &  9.9 &  3.8 \\ 
12 &  2.5 &  10.7 &  4.1 \\ 
13 &  2.9 &  9.0 &  4.5 \\ 
14 &  3.4 &  10.2 &  4.8 \\ 
15 &  3.8 &  10.1 &  5.1 \\
\end{tabular}
\hspace{0.7in}
% BINARY RANDOM de Bruij where k=2, based on 100,000 iterations
%
\begin{tabular} {r |  r  r   r }
 & \multicolumn{3}{c}{\footnotesize {\bf Ratio}: cover time / $2^n$} \\%(10,000 iterations)} \\
 $n$ & ~~~~~~~~~{\bf Min} & ~~~~~~~~~{\bf Max}  &~~~~~~~~~{\bf Avg}   \\  \hline
16 &  4.1 &  10.4 &  5.5 \\ 
17 &  4.5 &  9.8 &  5.9 \\ 
18 &  4.9 &  9.1 &  6.2 \\ 
19 &  5.3 &  9.8 &  6.5 \\ 
20 &  5.8 &  9.8 &  6.9 \\ 
21	& 6.1		& 9.4 	& 7.3 \\  %104 min
22	& 6.6		& 10.4 	& 7.8 \\   %210?  (laptop)
23	& 7.0		& 9.2 	& 8.0 \\   % 22h ? (laptop)
24	& 7.7		& 11.6 	& 9.1 \\   % 2479 min
25	& 7.6  	& 10.5 	& 8.5 \\
26	& 8.1		& 10.7	& 8.9 \\
%27	& 7.56	& 11.6 	& 10.1 \\
\end{tabular}
\end{center}
\caption{The minimum, maximum, and average ratios of the cover time to total edges in the de Bruijn graph $G(\mathbf{\Sigma}_2(n))$ by running Algorithm R for 10,000 iterations.}
\label{tab:DB}
\end{table}
%==============

%==============
\begin{table}
\begin{center}
\begin{tabular} {r |  r  r   r }
	& \multicolumn{3}{c}{\footnotesize {\bf Ratio}: cover time / $3^n$} \\ % (10,000 iterations)} \\
 $n$ & ~~~~~~~~~{\bf Min} & ~~~~~~~~~{\bf Max}  &~~~~~~~~ {\bf Avg}   \\  \hline
 3 & 0.3 & 3.9 & 0.9 \\
 4 & 0.5 & 4.5 & 1.3 \\
 5 & 0.7 & 4.6 & 1.7\\
 6 & 1.0 & 5.6 & 2.0\\
 7 & 1.5 & 5.1 & 2.4\\
 8 & 1.8 & 5.6 & 2.8\\
 9 & 2.2 & 6.3 & 3.1\\
 10 & 2.6 & 6.0 & 3.5 \\
 11 & 3.0 & 6.1 & 3.9\\
 12 & 3.4 & 6.3 & 4.3\\
13 & 4.0 & 6.4 & 4.6\\   % 7 hours
14 & 4.3  & 6.3 & 4.9\\   % 20 hours
15 &  4.6  & 6.6 & 5.3 \\    % cpu 93
16 &  5.0 &  6.7& 5.7 \\     %14934 m  -- 10 days
\end{tabular}
\hspace{0.7in}
\begin{tabular} {r |  r  r   r }
	& \multicolumn{3}{c}{\footnotesize {\bf Ratio}: cover time / $4^n$} \\ % (10,000 iterations)} \\
 $n$ & ~~~~~~~~~{\bf Min} & ~~~~~~~~~{\bf Max}  &~~~~~~~~ {\bf Avg}   \\  \hline
4 & 0.5 & 4.5 & 1.3 \\
 5 & 0.7 & 4.6 & 1.7\\
 6 & 1.0 & 5.6 & 2.0\\
 7 & 1.5 & 5.1 & 2.4\\
 8 & 1.8 & 5.6 & 2.8\\
 9 & 2.2 & 6.3 & 3.1\\
 10 & 2.6 & 6.0 & 3.5 \\
 11 & 3.0 & 6.1 & 3.9\\
 12 & 3.6 & 4.7 & 4.0\\  %1817 min
 13 &  3.9 &  5.0 &  4.3 \\   %7.5 days
\end{tabular}
\end{center}
\caption{The minimum, maximum, and average ratios of the cover time to total edges in the de Bruijn graph $G(\mathbf{\Sigma}_k(n))$ for $k=3$ and $k=4$ by running Algorithm R for 10,000 iterations.}
\label{tab:DB2}
\end{table}
%==============

If an application does not require a sequence generated uniformly at random, an algorithm which applies a Burrows-Wheeler transform can be applied; it outputs each de Bruijn sequence with positive probability~\cite{LP24}.  
%The algorithm also requires an exponential amount of time, with respect to $n$ and $k$, before the first symbol can be output and requires exponential space. 
\edit{The algorithm requires $\Theta(2^n)$ space and 
produces each symbol in $\mathcal{O}(\alpha(2^n))$ amortized time per symbol for $k=2$, where $\alpha(n)$ is the inverse Ackerman function. It is important to note that $\alpha(n)$ grows extremely slowly, with $\alpha(n)\leq 5$ for any $n$ of practical value. }
The first estimate on the mean \defo{discrepancy} of de Bruijn sequences is obtained using this algorithm~\cite{LP24}.

In Table~\ref{tab:DB-compare}, we compare  the running times of the (non-uniform) C++ algorithm implemented by the authors of~\cite{LP24} and our (uniform) algorithm implemented in C; both implementations are available online for download at~\cite{dbseq}. We average the running time over 10 iterations, not including the time it takes to output the sequence.\footnote{Our experiments were run on an iMac desktop with an Apple M4 processor. Our experimental times are faster than those reported in~\cite{LP24}, which is likely due to a faster processor.}  Interestingly, using the standard \texttt{rand()} function in the latter implementation resulted in a large cycle of generated bits, which did not allow all vertices to be visited during the ``random'' walk when $n \geq 28$. Thus, an alternate method for generating random bits had to be deployed.~\footnote{We added a value corresponding to the number of times the current vertex had been visited to the output of \texttt{rand()}.}  
If running time is a concern, the non-uniform result from~\cite{LP24} may be preferred for such applications. However, if a sequence is desired to be generated \emph{uniformly at random}, the one presented in this paper should be considered. It is also important to note that the algorithm described in this paper is a ``Las Vegas'' algorithm which means that no upper bound on the running time can be given, as it is not guaranteed to terminate.

%==============
\begin{table}
\begin{center}
\begin{tabular} {c |  r  r   }
 & \multicolumn{2}{c}{\footnotesize {\bf Average clock time in seconds} } \\
 $n$ & ~~~~~~~~~{\bf Non-uniform~\cite{LP24}} & ~~~~~~~~~{\bf Uniform}    \\  \hline
%20 &   $< 1$     &          \\
%21 &   $<$ 1   &  ~1   \\ 
22 &   $<$ 1 \ \ \ \ \ \ \ \    &  3 \ \  \ \\ 
23 &   $<$ 1 \ \ \ \ \  \ \ \  &  6  \ \ \  \\ 
24 &   $<$ 1 \ \ \ \ \  \ \ \   &  14 \ \  \ \\ 
25 &  2 \ \ \ \ \ \ \ \   &  31 \ \ \    \\ 
26	& 6 \ \ \ \ \  \ \ \  &  92 \ \  \   \\  
27	& 13 \ \ \ \ \ \ \ \  &  211 \ \  \   \\   
28	& 28 \ \ \ \ \  \ \ \   & 509	       \ \  \  \\
29	&  56 \ \ \ \ \ \ \ \  &  1094	    \ \  \  \\
30	&  122 \ \ \ \ \ \ \ \  &  2287 	   \ \  \  \\
\end{tabular}
\end{center}
\caption{Comparing the average clock time in seconds over 10 iterations between the algorithm from~\cite{LP24} which generates a de Bruijn sequence non-uniformly at random, and Algorithm R which generates a de Bruijn sequence uniformly at random.}
\label{tab:DB-compare}
\end{table}
%==============

%=====================================================================
\subsection{Weight-range de Bruijn sequences}  \label{sec:weight}
%=====================================================================

Let $\mathbf{WR}_k(n,[\ell,u])$ denote the subset of strings in $\mathbf{\Sigma}_k(n)$ with weight in the range $[\ell,u]$.  Let $\mathit{WR}_k(n,[\ell,u])$ denote $|\mathbf{WR}_k(n,[\ell,u])|$.
It is straightforward to observe that $\mathit{WR}_k(n,[\ell,u])=0$ if $u < 0$ or $\ell > n(k-1)$; otherwise, if $n=1$ then $\mathit{WR}_k(n,[\ell,u]) = min(u,k{-}1) - max(\ell,0) +1$, and if $n>1$:
\[\mathit{WR}_k(n,[\ell,u]) = \sum_{j=0}^{k-1} \mathit{WR}_k(n-1,[\ell{-}j,u{-}j]).\]
A \defo{weight-range de Bruijn sequence} is a universal cycle for the set $\mathbf{WR}_k(n,[\ell,u])$.  The de Bruijn graph $G(\mathbf{WR}_k(n,[\ell,u]))$ is generally not regular. A random edge $s_1s_2\cdots s_n$ can be generated using values for $\mathit{WR}_k(n,[\ell,u])$ as outlined in Algorithm~\ref{alg:weight}. The algorithm essentially unranks a string in $\mathbf{WR}_k(n,[\ell,u])$ as it appears in lexicographic order.

    \begin{lemma}
    Algorithm~\ref{alg:weight} can generate a string $s_1s_2\cdots s_n$ from $\mathbf{WR}_k(n,[\ell,u])$ uniformly at random using $\mathcal{O}(kn^3)$ simple operations on numbers up to $k^n$.
    \end{lemma}
    \begin{proof}
    Consider Algorithm~\ref{alg:weight}. 
    The required values $\mathit{WR}_k(n,[\ell,u])$, for $1 \leq j \leq n$ can be precomputed via dynamic programming using $\mathcal{O}(kn^3)$ simple operations on numbers up to $k^n$.  The outer {\bf for} loop iterates at most $n{-}1$ times and each iteration requires $\mathcal{O}(k+n)$ simple operations on numbers up to $k^n$.
    \end{proof}

\begin{comment}

\[ \mathit{WR}_k(n,[\ell,u]) = \begin{cases} 
      0 , &\text{if $u < 0$ or $\ell > n(k-1)$;}\\
      1 , & \text{if $n=0$;} \\
      0, & \text{if $n=1$ and $\ell > k-1$;} \\
     k-\ell, & \text{if $n=1$ and $u > k-1$;} \\
     u-\ell+1, & \text{if $n=1$ and $u \leq k-1$;} \\
       \sum\limits_{j=0}^{k-1} \mathit{WR}_k(n-1,[\ell - j,u-j]), & \text{otherwise.} \\
   \end{cases} \]

\[ \mathit{WR}_k(n,[\ell,u]) = \begin{cases} 
      1 , & \text{if $n=0$ or $[\ell, u] = [0,0]$;} \\
      \sum\limits_{j=0}^{\ell} \mathit{WR}_k(n-1,[\ell - j,u-j])+\sum\limits_{j=\ell+1}^{u} \mathit{WR}_k(n-1,[0,u-j]) , & \text{if $u < k-1$;} \\
      \sum\limits_{j=0}^{\ell} \mathit{WR}_k(n-1,[\ell - j,u-j])+\sum\limits_{j=\ell+1}^{k-1} \mathit{WR}_k(n-1,[0,u-j]) , & \text{if $u\geq k-1$ and $\ell < k-1$;} \\
      \sum\limits_{j=0}^{k-1} \mathit{WR}_k(n-1,[\ell - j,u-j]) , & \text{otherwise.} \\
   \end{cases} \]

\[\mathit{WR}_k(n,[\ell,u]) = \sum_{j=0}^{min(k-1,u)} \mathit{WR}_k(n-1,[max(\ell-j,0),u-j])\]
\end{comment}

%============================
\begin{algorithm}
    \small
    \begin{algorithmic}
  
    \State $r \gets $ random integer in  $[1,\mathit{WR}_k(n,[\ell,u])]$
   % \State $l \gets lower$
   % \State $u \gets upper$
    \For {$j$ {\bf from } $1$ {\bf to} $n-1$}
        \For{$i$ {\bf from} $0$ {\bf to} $k-1$} \ $n_i \gets  \mathit{WR}_k(n-j,[\ell-i,u-i])$ \EndFor
        \State $i \gets 0$
        \While{$r > n_i$}  \ 
            $r \gets r -  n_i$;  \ $i \gets i+1$
        \EndWhile
        \State $s_j \gets i$
        \State $\ell \gets \ell-i$; \  $u \gets u-i$  %\comm{update the bounds}
    \EndFor
    \If{$\ell < 0$}  \ $\ell \gets 0$  \EndIf
    \State $s_n \gets \ell+r-1$

    \end{algorithmic}
    \normalsize
    \caption{Random generation of a string $s_1s_2\cdots s_n$ in $\mathbf{WR}_k(n,[\ell,u])$  }
\label{alg:weight}
\end{algorithm}

  Table~\ref{tab:weight} shows the minimum, maximum, and average ratios of the cover time to total edges in $G(\mathbf{WR}_2(n,[5,10]))$ by running Algorithm R for 10,000 iterations.  

\begin{table}[h] 
\begin{center}
\begin{tabular} {r |  r  r   r }
	& \multicolumn{3}{c}{\footnotesize {\bf Ratio}: cover time / $\mathit{WR}_2(n,[5,10])$ } \\
 $n$ & ~~~~~~~~~{\bf Min} & ~~~~~~~~~{\bf Max}  &~~~~~~~~ {\bf Avg}   \\  \hline
 10 & 2.3  &   19.4  &  5.1 \\
 11 &2.5  &   15.9  &  5.5\\
 12 & 3.2  &   13.0  &  6.0\\  
 13 & 3.9  &   15.6  &  6.5\\
 14 & 4.2  &   15.8  &  7.2\\
 15 & 5.2  &   14.7  &  7.9\\ 
 16 & 6.0  &   16.0  &  8.6\\  %15min
 17 & 6.7  &   15.2  &  9.3 \\  %34 min
 18 & 7.4  &   17.0  &  10.2\\  
19 & 8.1  &   17.5  &  10.8\\   
 20 &  9.1  &   17.2  &  11.5 \\

\end{tabular}
\end{center}
\caption{The minimum, maximum, and average ratios of the cover time to total edges in the de Bruijn graph  $G(\mathbf{WR}_2(n,[5,10]))$ by running Algorithm R for 10,000 iterations.}
\label{tab:weight}
\end{table}

 Weight-range de Bruijn sequences can be constructed via an $\mathcal{O}(n)$ time per symbol successor rule when the minimum weight is 0, or the maximum weight is $(k-1)^n$ [GSWW20].
 %They can also be generated in $\mathcal{O}(1)$ amortized time per symbol using $\mathcal{O}(n^2)$ space by applying concatention trees~\cite{concattree}.
 In the binary case, they can be constructed for any weight range in $\mathcal{O}(1)$ amortized time [SWW13]. When $k=2$ and $\ell+1=u$, weight-range de Bruijn sequences correspond to the universal cycles for (shorthand) subsets discussed in Section~\ref{sec:perm}.

%=====================================================================
\subsection{de Bruijn sequences with forbidden $0^z$}  \label{sec:form}
%=====================================================================

A \defo{necklace class} is an equivalence class of strings under rotation; we call the lexicographically smallest string in the class a \defo{necklace}. The necklace class containing $\alpha$ is denoted $[\alpha]$.  For example, if $\alpha = 0001$ then $[\alpha] = \{0001, 0010, 0100, 1000\}$.
Let $\mathbf{N}_k(n,z)$ denote the set of all necklaces in $\mathbf{\Sigma}_k(n)$ with no $0^z$ substring for $z > 1$.  All such necklaces end with 1 when $z\leq n$.  Let $\displaystyle{\mathbf{Z}_k(n,z) = \bigcup_{\alpha \in \mathbf{N}_k(n,z)} [\alpha]}$.  It is known that $\displaystyle{\mathbf{Z}_k(n,z)}$ admits a maximal length universal cycle that does not contain the substring $0^z$~\cite{RLL2}.
We call a maximum length universal cycle that does not contain $0^z$ as a substring, a \defo{de Bruijn sequence with forbidden $0^z$}.

The de Bruijn graph $G(\mathbf{Z}_k(n,z))$ is not necessarily regular.  A random edge can be generated by applying the following recurrences.  Let $F_k(n,z)$ denote the number of $k$-ary strings
of length $n$ with no $0^z$ substring.  It satisfies the following recurrence for $z<n$:
\[ F_k(n,z) = (k{-}1)\sum_{j=1}^{z} F_k(n-j,z), \]
where $F_k(n,z) = k^n$ for $z > n$ and $F_k(n,n) = k^n-1$.

Let $Z_k(n,z)$ denote the number of $k$-ary strings of length $n$ with no $0^z$ substring, including the wraparound.  It satisfies the following recurrence for $z<n$ obtained by partitioning the strings into those beginning with a non-zero, and those with $j$ zeros in the wraparound, in which case there are $k{-}1$ possibilities for each of the first non-zero and last non-zero:
\[ Z_k(n,z) = (k{-}1)F_k(n-1,z) +  (k{-}1)^2\sum_{j=1}^{z-1} j \cdot F_k(n-j-2,z), \]
where $Z_k(n,z) = k^n$ for $z > n$ and $Z_k(n,n) = k^n-1$.

Given these recurrences, we can compute a random string in $\mathbf{F}_k(n,z)$ following a similar unranking strategy using lexicographic order as done with $\mathbf{WR}_k(n,[\ell,u])$ in the previous subsection. We omit the details in this case.
   Table~\ref{tab:forb} shows the minimum, maximum, and average ratios of the cover time to total edges in $G(\mathbf{F}_2(n,2))$ and $G(\mathbf{F}_2(n,3))$ by running Algorithm R for 10,000 iterations.  

\begin{table}[h] 
\begin{center}
\begin{tabular} {r |  r  r   r }
	& \multicolumn{3}{c}{\footnotesize {\bf Ratio}: cover time / $F_2(n,2)$ } \\
 $n$ & ~~~~~~~~~{\bf Min} & ~~~~~~~~~{\bf Max}  &~~~~~~~~ {\bf Avg}   \\  \hline
%

% 6 & 0.8 &  15.6 &  2.8\\
 %7 & 0.8 &  16.1 &  2.6\\
 8 & 1.0 &  20.9 &  3.6\\
 9 & 1.2 &  15.9 &  3.9\\   
 10 & 1.7 &  21.2 &  4.6 \\
 11 &2.0 &  13.5 &  4.8\\
 12 & 2.4 &  17.6 &  5.7\\  
 13 & 2.9 &  15.7 &  5.7\\
 14 & 3.3 &  19.0 &  6.4\\
 15 & 3.8 &  18.1 &  6.9\\ 
 16 & 4.3 &  21.9 &  7.4\\
 17 & 4.3 &  15.9 &  7.8 \\
 18 & 5.2 &  21.9 &  8.4\\  
19 & 5.9 &  18.7 &  8.8\\   
 20 & 6.4 &  17.2 &  9.3\\ 
 21 & 6.9 &  21.2 &  9.7 \\
 22 &  7.6 &  20.1 &  10.2 \\ %20 min
 23 & 8.0 &  18.7 &  10.7 \\
 24 & 8.3 &  24.8 &  11.2 \\ %60min

\end{tabular} \ \ \ \ \ \ \ \  \ \ \ \ \ \ \ \ \
\begin{tabular} {r |  r  r   r }
	& \multicolumn{3}{c}{\footnotesize {\bf Ratio}: cover time / $F_2(n,3)$ } \\
 $n$ & ~~~~~~~~~{\bf Min} & ~~~~~~~~~{\bf Max}  &~~~~~~~~ {\bf Avg}   \\  \hline
 8 & 1.4 &  21.5 &  4.3\\
 9 & 1.8 &  13.3 &  4.6\\   
 10 & 2.5 &  18.7 &  5.3 \\
 11 & 3.1 &  14.8 &  5.8\\
 12 & 3.8 &  17.5 &  6.5\\  
 13 & 4.1 &  16.5 &  7.0\\
 14 & 4.5 &  17.0 &  7.7\\
 15 & 5.5 &  16.7 &  8.2\\ 
 16 & 6.1 &  20.0 &  8.9\\
 17 & 6.7 &  19.8 &  9.5 \\
 18 & 7.4 &  17.7 &  10.1\\  %22min
19 & 8.2 &  16.3 &  10.6\\   %48
 20 & 8.8 &  16.9 &  11.3\\  %95  
 21 & 9.6 &  15.8 &  11.8  \\  %3h
 22 &  10.2 &  17.7 &  12.5 \\
 23 &  10.8 &  18.0 &  13.1 \\
 24 &  11.0 &  22.8 &  13.9 \\
\end{tabular}
\end{center}
\caption{The minimum, maximum, and average ratios of the cover time to total edges in the de Bruijn graphs $G(\mathbf{F}_2(n,2))$ (left) and $G(\mathbf{F}_2(n,3))$ (right) by running Algorithm R for 10,000 iterations.}
\label{tab:forb}
\end{table}

 The lexicographically smallest de Bruijn sequences with forbidden $0^z$ can be generated via a simple greedy algorithm~\cite{generalize-classic-greedy}; it can also be generated efficiently by concatenating the aperiodic prefixes of necklaces with no $0^z$ substring as they appear in lexicographic order~\cite{generalize-classic-greedy,smaller_universal_cycles}.  An exponential number of such sequences can be efficiently generated by applying cycle-joining as described in~\cite{RLL2}.

%=====================================================================
\subsection{Orientable sequences}  \label{sec:orient}
%=====================================================================

%Recall that a \defo{necklace class} is an equivalence class of strings under rotation; we call the lexicographically smallest string in the class a \defo{necklace}. The necklace class containing $\alpha$ is denoted $[\alpha]$.  

Recall the definitions of a necklace class and necklace from the previous subsection.
A \defo{bracelet class} is an equivalence class of strings under rotation and reversal; we call the lexicographically smallest string in the class a \defo{bracelet}.
A \defo{bracelet} is said to be \defo{asymmetric} if it is not in the same necklace class as its reversal. For example, $001011$ is an asymmetric bracelet, but $001001$ is not. 
Let $\mathbf{AB}_k(n)$ denote the set of all $k$-ary asymmetric bracelets of length $n$, and let $\mathbf{OS}_k(n) = \bigcup_{\alpha \in \mathbf{AB}_k(n)} [\alpha]$.  Let $OS_k(n)$ denote $|\mathbf{OS}_k(n)|$.

%\begin{sloppypar}
For example, $\mathbf{AB}_2(7) =  \{0001011, 0010111\}$, and  
\begin{eqnarray*}   \small
\mathbf{OS}_2(7) & = & \{ 0001011, 0010110, 0101100, 1011000, 0110001, 1100010, 1000101\} \cup \\
& & \{ 0010111, 0101110, 1011100, 0111001, 1110010, 1100101, 1001011\}, 
\end{eqnarray*}
where $OS_2(7) = 14$.
%\end{sloppypar}
%

An \defo{orientable sequence} is cyclic sequence such that each length-$n$ substring occurs at most once in \emph{either direction}.  For example, a maximum-length orientable sequence for $n=5$ and $k=2$ is $001101$.
A universal cycle for $\mathbf{OS}_k(n)$ is known to be an orientable sequence with asymptotically optimal length~\cite{BM,kary-orientable}.
%Consider the induced subgraph $H$ of the de Bruijn graph $G(\mathbf{\Sigma}_k(n))$ where each vertex has an outgoing edge corresponding to
%a string in $\mathbf{OS}_k(n)$.  
A formula for $OS_k(n)$, is given in~\cite{Dai,korient}. 
Generating a random string from $\mathbf{OS}_k(n)$ does not appear to be a trivial matter.  However, by randomly generating $k$-ary strings with rejection, on average only two random strings need to be generated to obtain a string in $\mathbf{OS}_k(n)$ as $n$ gets large.  Thus, the expected time to generate a random edge in $G(\mathbf{OS}_k(n))$ is $\Theta(n)$.

Table~\ref{tab:orient} illustrates the minimum, maximum, and average ratios of the cover time to total edges in $G(\mathbf{OS}_2(n))$ by running Algorithm R for 10,000 iterations. 

\begin{table}
\begin{center}
\begin{tabular} {r |  r  r   r }
	& \multicolumn{3}{c}{\footnotesize {\bf Ratio}: cover time / $OS_2(n)$} \\ % (10,000 iterations)} \\
 $n$ & ~~~~~~~~~{\bf Min} & ~~~~~~~~~{\bf Max}  &~~~~~~~~ {\bf Avg}   \\  \hline

 6 & 0.8 & 0.8 & 0.8\\
 7 & 0.9 & 8.5 & 1.5\\
 8 & 0.9 & 18.8 & 3.2\\
 9 & 1.3 & 13.2 & 4.3\\   
 10 & 2.1 & 14.7 & 5.0 \\
 11 & 2.7 & 18.4 & 5.7\\
 12 & 3.6 & 16.3 & 6.5\\  
 13 & 4.2 & 17.0 & 7.2\\
 14 & 5.1 & 17.6 & 7.9\\
 15 & 5.3 & 18.0 & 8.5\\ %40 min
 16 & 6.2& 18.9 & 9.2\\
 17 & 7.0 & 18.4 & 9.9\\
 18 & 7.8 & 18.6 & 10.6\\  %4.5 h 
 19 & 8.6 & 17.7 & 11.2\\   %9 h
 20 & 9.5 & 17.7 & 11.8\\  %  20 h

\end{tabular}
\end{center}
\caption{The minimum, maximum, and average ratios of the cover time to total edges in the de Bruijn graph $G(\mathbf{OS}_2(n))$ by running Algorithm R for 10,000 iterations.}
\label{tab:orient}
\end{table}

Orientable sequences with asymptotically optimal length can be constructed in $\mathcal{O}(n)$ time per symbol using $\mathcal{O}(n)$ space~\cite{korient}; in the binary case, they can be constructed in $\mathcal{O}(1)$ amortized time per bit using $\mathcal{O}(n^2)$ space.

\bibliographystyle{acm.bst}
\bibliography{refs}

%==================================
%==================================
%==================================
%\appendix

%\newpage
%\section{C code}  \label{app:code}

%\scriptsize
%\begin{code}
%\end{code}
\end{document}